\documentclass[11pt]{amsart}
\usepackage{amsthm}

\usepackage[all]{xy}
\usepackage{amssymb}
\usepackage{enumerate}
\usepackage{mathrsfs}
\usepackage{epsfig}
\usepackage{graphicx}
\usepackage{subfig}
\usepackage{float}
\usepackage{epigraph}

\numberwithin{equation}{section}

\newtheorem{thm}{Theorem}[section]

\newtheorem{lem}[thm]{Lemma}

\newtheorem{rem}{Remark}[section]
\newtheorem{example}[thm]{Example}

\newcommand{\eq}[1]{(\ref{#1})}

\renewcommand{\Re}{\operatorname{\rm Re}}
\renewcommand{\Im}{\operatorname{\rm Im}}

\newcommand{\beqast}{\begin{eqnarray*}}
\newcommand{\eqast}{\end{eqnarray*}}
\newcommand{\beqa}{\begin{eqnarray}}
\newcommand{\eqa}{\end{eqnarray}}

\newcommand{\bbe}{\begin{equation}}
\newcommand{\ee}{\end{equation}}

\renewcommand{\Re}{\operatorname{\rm Re}}
\renewcommand{\Im}{\operatorname{\rm Im}}

\newcommand{\bE}{{\mathbb E}}

\newcommand{\bQ}{{\mathbb Q}}

\newcommand{\bR}{{\mathbb R}}
\newcommand{\bC}{{\mathbb C}}
\newcommand{\bZ}{{\mathbb Z}}

\newcommand{\cA}{{\mathcal A}}

\newcommand{\cF}{{\mathcal F}}

\newcommand{\cD}{{\mathcal D}}

\newcommand{\cS}{{\mathcal S}}

\newcommand{\cL}{{\mathcal L}}

\newcommand{\cU}{{\mathcal U}}

\newcommand{\hG}{{\hat G}}
\newcommand{\hV}{{\hat V}}

\newcommand{\al}{\alpha}
\newcommand{\be}{\beta}

\newcommand{\de}{\delta}

\newcommand{\ka}{\kappa}
\newcommand{\la}{\lambda}
\newcommand{\lp}{\lambda_+}
\newcommand{\lm}{\lambda_-}
\newcommand{\La}{\Lambda}

\newcommand{\sg}{\sigma}

\newcommand{\om}{\omega}

\newcommand{\ga}{\gamma}

\newcommand{\Ga}{\Gamma}

\newcommand{\dd}{\partial}

\newcommand{\hu}{{\hat u}}

\newcommand{\bfo}{{\bf 1}}

\begin{document}

\title[Gauge transformations in the dual space]
{Gauge transformations in the dual space, and pricing and estimation in the long run  in affine
 jump-diffusion models}

\author[
Svetlana Boyarchenko and
Sergei Levendorski\u{i}]
{
Svetlana Boyarchenko and
Sergei Levendorski\u{i}}

\begin{abstract}
We suggest a simple reduction of pricing European options in 
affine jump-diffusion
models to pricing options with modified payoffs in diffusion models. The procedure is based
on the conjugation of the infinitesimal generator of the model with an operator of the form $e^{i\Phi(-i\dd_x)}$
(gauge transformation in the dual space). A general procedure for
the calculation of the function $\Phi$ is given, with examples. As applications, we consider pricing 
 in jump-diffusion models and their subordinated versions using the eigenfunction expansion technique, and estimation of the extremely rare jumps component. The beliefs of the market about yet unobserved extreme jumps
and pricing kernel
can be recovered: the market prices allow one to see  ``the shape of things to come".



\end{abstract}

\thanks{
\emph{S.B.:} Department of Economics, The
University of Texas at Austin, 2225 Speedway Stop C3100, Austin,
TX 78712--0301, {\tt sboyarch@eco.utexas.edu} \\
\emph{S.L.:}
Calico Science Consulting. Austin, TX.
 Email address: {\tt
levendorskii@gmail.com}}

\maketitle




\tableofcontents

\section{Introduction}
The crucial observation, which makes the the Heston model  \cite{heston-model}, 
affine term structure models (ATSMs)  \cite{DPS,DFS}, Barndorff-Nielsen-Shepard model
\cite{BNsh}, and other affine models very tractable, hence, popular,
is that, in these models, the characteristic function of the transition density is of the form
of an exponential function of an affine function of factors of the models with
the coefficients depending on time to maturity $\tau=T-t$  and spectral parameter $\xi$:
\begin{equation}\label{chexpaff0}
\bE_t^x[e^{i\langle\xi, X_T\rangle}]=\exp[\langle A(\tau,\xi),x\rangle+B(\tau, \xi)].
\end{equation}
Here and below, $\langle \cdot, \cdot\rangle$ denotes the standard inner product in $\bR^n$, and $i:=\sqrt{-1}$
is the imaginary unit. If the
stochastic interest rate $r_t$ is modelled as an affine function of the factors of the model,
the state space must be enlarged as in the probabilistic version of the Feynman-Kac
formula, and an additional factor $Y_t=\int_0^t r_sds$ added. See \cite{DPS,DFS} for details. The extended model
remains affine, and a representation of the characteristic function
in the form  \eqref{chexpaff0} becomes possible.

Typically, in applications, the dynamics of the factors is given by a system of stochastic differential equations (SDE). Hence,
 the crucial equality \eqref{chexpaff0} has to be proved. The formal proof \cite{DPS} is straightforward. Denote by $V(t,x)$ the LHS of \eqref{chexpaff0},
and by $\cL$ the infinitesimal generator of $X$,  write down the Cauchy problem for the backward Kolmogorov equation
\begin{eqnarray}\label{bK}
(\dd_t+\cL)V(t,x)&=&0,\quad x\in \cD_0,\ t\in [0,T),\\\label{bcbK}
V(T,x)&=&e^{i\langle x, \xi\rangle},
\end{eqnarray}
where $\cD_0$ is the interior of the state space $\cD$, substitute the anzatz \eqref{chexpaff0} into \eqref{bK}--\eqref{bcbK},
and, using the fact that, in affine models, the coefficients of the diffusion part of $\cA$ and L\'evy measure are affine functions of the
factors, deduce that the vector-function $F(\tau, \xi)= (A(\tau,\xi), B(\tau, \xi))$ can be found as the solution
of a Cauchy problem for the system of generalized Riccati equations associated with the model.

From the pure mathematical point of view, pricing European options in affine models is straightforward:
calculate the (conditional) characteristic function solving the associated system of generalized Riccati equations,
and apply the inverse Fourier transform to calculate the option price. See \cite{DFS,pitfalls} and the bibliographies therein.
In several popular affine diffusion models,
the associated system of generalized Riccati equations can be solved analytically, and then numerical 
realizations of the pricing problem simplify significantly. However,
even in this case, an accurate and fast option pricing is not trivial, and serious errors can result, for options of short and long maturity
especially. See \cite{HestonCalibMarcoMe,HestonCalibMarcoMeRisk} for examples in the context of calibration of the Heston model,
and \cite{paraHeston,SINHregular} for efficient numerical integration procedures. In the case of models of the CIR type,
difficulties for accurate calculations are even more serious \cite{pitfalls,SINHregular}. If jumps are added, then, with rare exceptions,
an analytic solution of the associated system of Riccati equations is impossible, and one has to solve the system numerically.
For large values of the spectral parameter and for options of long and short maturity especially, an accurate solution
is very difficult. In particular, a numerical solution can even blow up although the exact solution does not. These effects are studied
in  \cite{pitfalls}.

An additional advantage of affine and quadratic diffusion models is that 
options of long maturity can be efficiently priced using the eigenfunction 
expansion method (see, e.g., \cite{DavidovLinetsky,GorovoiLinetsky,LiLinetsky2014a,LimLiLinetsky2012} for examples
of application of the eigenfunction expansion technique in one-factor models, and \cite{ninaEigenOU,ninaEigenMF,BarrStIR} in multi-factor models, non-diagonalizable ones
including).  The pricing in the subordinated models is also easy if the eigenfunction expansion approach is applicable \cite{ninaEigenOU,ninaEigenMF,LiLinetsky2014a}.
In \cite{ninaAsymp}, options of long maturity in affine and quadratic models with jumps are priced using elements of asymptotic analysis
and eigenfunction expansion technique. 

In the present paper, we suggest a simple trick which, in effect, allows one to reduce pricing in many affine jump-diffusion models
to pricing in the same model but with the jump component eliminated. Hence, the eigenfunction expansion technique becomes
applicable.\footnote{In a separate paper, we will construct a version of the trick for quadratic jump-diffusion models,
and consider the same applications as in the present paper.}  A different applications of the eigenfunction expansion technique
in a special case of the so-called basic affine jump diffusion\footnote{CIR model with jumps; the jump sizes are exponentially distributed} can be found in \cite{LiMendozaMitchell2016}.

 Let $\cL=\cL_{\mathrm{diff}}+\cL_{\mathrm{jump}}$ be the infinitesimal generator
of the underlying process, the (stochastic) discounting being taken into account, let $x$ be the vector of the state variables
and denote $D=-i\dd_x(=-i\nabla_ x)$. We conjugate the infinitesimal generator
 $\cL$ with an operator of the form $e^{i\Phi(D)}$. \footnote{Recall that {\em the pseudo-differential operator} $A=A(x,D)$
 with symbol $a(x,\xi)$ is defined by $A=\cF^{-1}_{\xi\to x}a(x,\xi)\cF_{x\to\xi}$, where $\cF$ is the Fourier transform. See, e.g.,
\cite{NG-MBS}.} The function $\Phi$ is chosen so
that
\bbe\label{gauge1}
\cL_{\mathrm{diff}}+\cL_{\mathrm{jump}}=e^{i\Phi(D)}L_{\mathrm{diff}}e^{-i\Phi(D)}
\ee
as operators acting on functions of interest.
Let $P^{JD}_t=\exp[t(\cL_{\mathrm{diff}}+\cL_{\mathrm{jump}})]$ and $P^D_t=\exp[t(\cL_{\mathrm{diff}})]$ be transition operators in the affine jump-diffusion model
and affine diffusion model, respectively. Let $G(X_T)$ be the payoff of the European option at maturity date $T$.
Then the option price at time 0, in the jump-diffusion model, can be represented in the form
\bbe\label{EuroJD}
V(T,x)=P^{JD}_TG(x)=e^{i\Phi(D)}P^D_Te^{-i\Phi(D)}G.
\ee
The advantage of the multiplicative structure of the right-most side of \eq{EuroJD} is that 
the action of $P^D_T$ is easier  to realize accurately than the action of $P^{JD}_T$ if the standard iFT (inverse Fourier transform method) is applied 
(for pricing options of long maturity, an accurate numerically realization of $P^{JD}_T$ can be extremely difficult), and
the direct application of the eigenfunction expansion method to $P^{JD}_T$ is impossible whereas the eigenfunction
expansion method is applicable in many models, multi-factor Ornstein-Uhlenbeck driven models including \cite{ninaEigenOU,ninaEigenMF}. The disadvantage is an additional problem of an accurate calculation of $\Phi(\xi)$ at points of the grid in the dual space, used
in a chosen numerical Fourier transform method; however, a numerical solution of this problem is simpler 
than the numerical solution of the system of Riccati equations in the jump-diffusion case, for large values of the
spectral parameter and long and short maturities especially. 

The conjugation with the multiplication operators of the form $e^{i\varphi(x)}$ (gauge transformations) is the standard tool
in the theory of electric and magnetic fields; as above, the goal is to transform the Hamiltonian into a more convenient form.
In probability, the conjugation with functions of the form $e^{\langle a,  x\rangle}$ is interpreted as the change of measure.
For a particular choice of $a$, this is the well-known Esscher transform in the one-factor case; 
for L\'evy processes in $\bR^n$, the generalization of the Esscher transform is introduced in \cite{NG-MBS}. 
In \cite{ninaEigenOU,ninaEigenMF}, the conjugation with operators of the form $e^{\Psi(x)}e^{\langle a, x\rangle}e^{\langle b, D\rangle}$, 
where $a,b$ are constants and
$\Psi(x)$ is a quadratic form, is used
to diagonalize the infinitesimal generator in multi-factor models of Ornstein-Uhlenbeck type (affine and quadratic diffusion models)
or, in the essentially non-self-adjoint case, to reduce to the case of the infinite direct sum of Jordan blocks.
Thus, the novelty of the present paper is to use a more general functions of $D$ than exponentials of linear functions 
$e^{\langle b, D\rangle}$ to establish the essential equivalency of a certain class
of jump-diffusion models and the corresponding class of pure diffusion models, as far as the pricing European options
is concerned.  One may expect that conjugation with pseudo-differential operators of a more general form and more involved
Fourier Integral Operators technique  can
 be used to establish the equivalency of wider classes of jump-diffusion models. 


The plan of the paper is as follows. In Sect.~\ref{st_scheme}, we recall the standard scheme of
pricing in affine jump-diffusion models. 
The method of gauge transformations in the dual space in applications to affine jump-diffusion models is in Sect.~\ref{elimination}.
If the jump component can be completely eliminated and eigenvalues and a basis consisting of eigenfunctions 
\footnote{and generalized eigenfunctions, if the infinitesimal generator of the diffusion process is
essentially non-selfadjoint} can be explicitly calculated then the eigenfunction expansion technique is applicable.
We consider examples of such models and their subordinated versions in Sect.~\ref{Eigenexpansions}.
In  Sect.~\ref{estimation}, the elimination procedure is used to identify the 
 extremely rare jumps  component under the risk-neutral measure chosen by the market even if
 these rare jumps were not observed in the past: ``the shape of things to come" as seen by the market.
Sect.~\ref{concl} concludes.


\section{Standard scheme of option pricing in affine models}\label{st_scheme} 
Let $\cD$ be the state space. We restrict ourselves to the cases $\cD=\bR^m$ and $\cD=\bR^{m'}_+\times\bR^{m-m'}$,
where $0\le m'\le m$. For $a,b\in \bC^m$, 
set $\langle a, b\rangle =\sum_{j=1}^ma_jb_j$,
and
consider the European option with the payoff $G(X_T)$ at the maturity date $T$.  Our first standing assumption is

\vskip0.1cm
\noindent 
{\sc Condition $\hG$.} There exists an open set $U_G\in \bR^m$ s.t. the Fourier transform 
\bbe\label{defhG}
\hG(\xi)=\int_{\cD} e^{-i\langle x,\xi\rangle}G(x)dx
\ee
  is well-defined on a tube domain $\cU_G=iU_G+\bR^m\subset \bC^m$, and  $\forall \om\in U_G$,
$\hG(i\om+\cdot)\in L_1(\bR^m)$.

 Let $ChF(T,x,\xi)=\bE^{\bQ,x}\left[e^{i\langle \xi, X_T\rangle}\right]$
be the characteristic function of $X_T$ conditioned on $X_0=x\in \cD$, under the risk-neutral measure
$\bQ$ chosen for pricing. In an affine model,  $ChF(T,x,\xi)$  is of the form
\[
ChF(T,x,\xi)=\exp[\langle A(T,\xi),x\rangle +B(T,\xi)],
\]
where $A(t,\xi)$ is a vector-function, $A(0,\xi)=i\xi$, $B(t,\xi)$ is a scalar function, and the 
vector-function $[A(t,\xi)\ B(t,\xi)]$ is 
the solution of the system of generalized Riccati equations associated with the model (see \cite{DFS}).
To be more specific, if the infinitesimal generator of the model is of the form
\bbe\label{infgen}
\cL=\langle x, L(D) \rangle +L_0(D),
\ee
where $L(D)=[L_j(D)]_{j=1}^m$, then $A$ is the solution of the system 
\bbe\label{eq:A}
\frac{dA_{j}(t,\xi)}{d\tau}=L_j(-iA(t,\xi)),\ j=1,\ldots, m,
\ee
subject to $A(0,\xi)=i\xi$, and the function $B$ is found by integration
\bbe\label{eq:B}
B(t,\xi)=\int_0^t L_0(-iA(s,\xi))\,ds.
\ee
\vskip0.1cm
\noindent 
{\sc Condition $ChF(T)$.} There exist  tube domains $\cU_{ChF}=iU_{ChF}+\bR^m\subset \bC^m$ and 
$\cU^+_{ChF}=iU^+_{ChF}+\bR^m\subset \bC^m$,
where $U_{ChF}\subset U^+_{ChF}\subset \bR^m$ are open and non-empty, such that
\begin{enumerate}[(i)]
\item
$L_j, j=0,1,\ldots, m$, are analytic in $\cU^+_{ChF}$;
\item
for any $\xi\in \cU_{ChF}$ and any $s\in [0,T]$, $-iA(s,\xi)\in \cU^+_{ChF}$;
\item
there exists $C>0$ s.t. $\Re B(T,\xi)<C$, $\forall$\ $ \xi\in \cU_{ChF}$.
\end{enumerate}
Assume that $U_{ChF}\cup U_G\ne \emptyset$. Then we take $\om\in U_{ChF}\cup U_G$,
substitute 
\bbe\label{iFThG}
G(X_T)=(2\pi)^{-m}\int_ {i\om+\bR^m}e^{i\langle X_T,\xi\rangle} \hG(\xi)d\xi
\ee
into 
$
V(T,x)=\bE^{\bQ,x}\left[G(X_T)\right]$, use Fubini's theorem to justify the change of the order of taking integration and expectation,
and derive 
\bbe\label{price0}
V(T,x)=(2\pi i)^{-m}\int_{i\om+\bR^m} e^{\langle A(T,\xi),x\rangle +B(t,\xi)}\hG(\xi)d\xi.
\ee
\begin{rem}{\rm \begin{enumerate}
[(1)]
\item
In many cases of interest, $G$ depends only on some of the state variables
or a linear combination of state variables; then, as a function of $\xi\in\bR^m$, $\hG$ is a generalized function,
which is a tensor product of an appropriate Dirac delta function and regular function.
The integral reduces to an integral of a lower dimension, in many cases, one-dimensional integral.
\item
The scheme is simple but if the subsystem \eq{eq:A} does not admit an explicit analytical solution
then a sufficiently accurate solution of equation \eq{eq:A} and accurate calculation of integrals \eq{eq:B} and \eq{price0}
are very difficult, for options of long maturities especially. See \cite{pitfalls} for an analysis of these difficulties.
\item
The main source of difficulties is a very irregular behavior of the RHS of \eq{eq:A} for $\xi$ large in absolute value.
At the same time, typically, $\hG$ does not decay fast at infinity, and in models of CIR type, for any $T$,
the characteristic function decays slowly as $\xi\to\infty$. Hence, one must numerically evaluate the integral \eq{price0}
taking into account the values of the integrands at points $\xi$ of large absolute value; these values are very difficult to
calculate accurately. See \cite{pitfalls,SINHregular} for examples.
 \item

The reduction of  a jump-diffusion model to the corresponding diffusion model and application of the eigenfunction expansion technique 
is especially advantageous  if the Fourier transforms of the eigenfunctions decay very fast at infinity.
For instance, this is the case in models of the Ornstein-Uhlenbeck type: the Fourier transform of the eigenfunctions decay
as $\exp(-c|\xi|^2)$, where $c>0$, hence, it is sufficient to do accurate calculations for $\xi$ of a moderate absolute value. 

\item
In many cases, Conditions $\hG$ and $ChF$ are valid with
unions of tube domains and coni in place of tubes domains, and then it is advantageous to deform the flat surface of integration
in \eq{price0} into a surface stabilizing at infinity to an appropriate cone. See \cite{pitfalls,SINHregular} for details.

\end{enumerate}
}\end{rem}


\section{Elimination of a jump component in affine jump-diffusion models}\label{elimination}
\subsection{Elimination-estimation equation (EE-equation)}
Let the infinitesimal generator of the model be of the form
\bbe\label{infgen0}
\cL=\langle x, L(D) \rangle +L_0(D)+L_J(D),
\ee
where $L(D)=[L_j(D)]_{j=1}^m$, and $L_J(D)$ is the component which we want to eliminate.
If $L_j(D), j=0,1,\ldots, m$, are differential operators, then the dual gauge transformation constructed below
can be used to reduce pricing in affine jump-diffusion models to pricing in affine diffusion models. 
We assume that Conditions $\hG$ and $ChF$ hold.

Define 
\bbe\label{cL0}
\cL^0=\langle x, L(D) \rangle +L_0(D).
\ee
We need to find a function $\Phi$ such that
that
\bbe\label{gauge2}
\cL^0+L_J(D)=e^{i\Phi(D)}\cL^0e^{-i\Phi(D)}
\ee
as operators acting on functions of interest.
Since $\cF x\cF^{-1}=i\nabla_\xi$, and $\cF D_j\cF^{-1}=\xi_j$, an equivalent 
dual form of \eq{cL0} is 
\bbe\label{dual_infgen}
\cF \cL^0\cF^{-1}=\sum_{j=1}^m i\dd_{\xi_j}L_j(\xi)+L_0(\xi)
=\sum_{j=1}^m iL_j(\xi)\dd_{\xi_j}+L_{0}(\xi)+L_{01}(\xi),
\ee
where $ L_{01}(\xi)=i\sum_{j=1}^m \dd_{\xi_j} L_j(\xi)$.
Multiplying \eq{dual_infgen} on the right by $e^{-i\Phi(\xi)}$, and on the left by $e^{i\Phi(\xi)}$, 
and using 
$
e^{i\Phi(\xi)}i\dd_{\xi_j}e^{-i\Phi(\xi)}=\Phi_{\xi_j}(\xi)+i\dd_{\xi_j}$, we obtain
\bbe\label{Ldecomp3}
e^{i\Phi(\xi)}\cF \cL^0\cF^{-1}e^{-i\Phi(\xi)}=\sum_{j=1}^m iL_j(\xi)\dd_{\xi_j}+L_{0}(\xi)+L_{01}(\xi)+
\sum_{j=1}^m L_j(\xi)\dd_{\xi_j}\Phi(\xi).
\ee
Hence, to satisfy \eq{gauge2}, we need to find a $\Phi$ satisfying the following first order PDE
\bbe\label{eqPhi1}
\sum_{j=1}^m L_j(\xi)\dd_{\xi_j}\Phi(\xi)=L_{J}(\xi).
\ee
In Sect.~\ref{Eigenexpansions} and \ref{estimation}, we will use \eq{eqPhi1} in two ways:
\begin{enumerate}[(1)]
\item
if the parameters of $\cL$ are known, and efficient eigenfunction expansion for $\cL^0$
is available, we use \eq{eqPhi1} to eliminate the component $L_J$ and apply the eigenfunction expansion to the model
with the infinitesimal generator $\cL^0$;
\item
if the parameters of $\cL^0$ are inferred from the prices of options of short maturities, we use the prices of options of long
maturities to find $\Phi$, and then \eq{eqPhi1} to calculate $L_J$. In the multi-factor case, it may be necessary
to make a preliminary change of measure (see Theorem.~\ref{thm_Phi_gen}). 
\end{enumerate}
\subsection{The case $m=1$} 
The solution is straightforward. Let 
 $L_1(\xi)\neq 0,\ \forall\  \xi\in \cU_{ChF}$. 
We fix $a\in\cU_{ChF}$, and, for and $\xi\in \cU_{ChF}$, define
\bbe\label{Phimeq1}
\Phi(a;\xi)=\int_{\Ga(a,\xi)}\frac{L_{\mathrm{jump}}(\eta)}{L_1(\eta)}d\eta,
\ee
where $\Ga(a,\xi)\subset \cU_{ChF}$ is an arbitrary regular contour in the imply connected domain of interest. 
For any $a, a'$ in such a domain,
$\Phi(a,\cdot)-\Phi(a',\cdot)$ is a constant function, hence, the choice of $a$ is not important.
Several examples:


\subsubsection{Ornstein-Uhlenbeck model, $m=1$}\label{ss:Oumeq1}
The infinitesimal generator is 
\bbe\label{eq:OUmeq1}
\cL=\frac{\sg^2}{2}\dd^2_x+\ka(\theta-x)\dd_x+L_{\mathrm{jump}}(D),
\ee
hence, 
\bbe\label{OUPhimeq1}
\Phi(a;\xi)=\frac{1}{\ka}\int_{\Ga(a,\xi)}\frac{L_{\mathrm{jump}}(\eta)}{i\eta}d\eta.
\ee
Let $L_{\mathrm{jump}}(\xi)$ be analytic in a strip $(\lm,\lp)$, where $\lm<-1<0<\lp$. Then,
to price put options, we take $a, \xi$ and a contour $\Ga(a,\xi)$ in the strip $\{\Im\xi\in (0,\lp)\}$,
and to price call options, in the strip $\{\Im\xi\in (\lm,-1)\}$. In simple cases such as the double-exponential jumps
\bbe\label{DEJD}
L_J(\xi)=\frac{c_+i\xi}{\lp+i\xi}+\frac{c_-i\xi}{-\lm-i\xi},
\ee
and in cases when $L_J$ becomes a rational function after a polynomial change of variables,
the integral on the RHS of \eq{OUPhimeq1} can be calculated explicitly; in the general case, for $\xi$ of interest,
$\Phi(a;\xi)$ has to be calculated numerically.

\subsubsection{Vasicek model}\label{ss:Vasmeq1}
The infinitesimal generator is 
\bbe\label{eq:Vasmeq1}
\cL=\frac{\sg^2}{2}\dd^2_x+\ka(\theta-x)\dd_x+L_{\mathrm{jump}}(D)-x,
\ee
hence, 
\bbe\label{VasPhimeq1}
\Phi(a;\xi)=\int_{\Ga(a,\xi)}\frac{L_{\mathrm{jump}}(\eta)}{i\ka\eta-1}d\eta.
\ee
Let $L_{\mathrm{jump}}(\xi)$ be analytic in a strip $(\lm,\lp)$, where $\lm<-1<0<\lp$. Then,
to price put options, we take $a, \xi$ and a contour $\Ga(a,\xi)$ in the strip $\{\Im\xi\in (0,\lp)\}$,
and to price call options, in the strip $\{\Im\xi\in (\lm,-1)\}$. In simple cases such as double-exponential jumps
\eq{DEJD} 
and in cases when $L_J$ becomes a rational function after a polynomial change of variables,
the integral on the RHS of \eq{VasPhimeq1} can be calculated explicitly; in the general case, for $\xi$ of interest,
$\Phi(a;\xi)$ has to be calculated numerically.

\subsubsection{Square root model}\label{ss:Sqmeq1}
The infinitesimal generator is 
\bbe\label{eq:Sqmeq1}
\cL=\frac{\sg^2}{2}x\dd^2_x+\ka(\theta-x)\dd_x+L_{\mathrm{jump}}(D),
\ee
hence, 
\bbe\label{SqPhimeq1}
\Phi(a;\xi)=\int_{\Ga(a,\xi)}\frac{L_{\mathrm{jump}}(\eta)}{-\sg^2\eta^2/2+i\ka\eta}d\eta.
\ee
Let $L_{\mathrm{jump}}(\xi)$ be analytic in a strip $(\lm,\lp)$, where $\lm<-1<0<\lp$. Then,
to price put options, we take $a, \xi$ and a contour $\Ga(a,\xi)$ in the strip $\{\Im\xi\in (0,\lp)\}$,
and to price call options, in the strip $\{\Im\xi\in (\lm,-1)\}$. In simple cases such as exponential jumps
and in cases when $L_J$ becomes a rational function after a polynomial change of variables,
 the integral on the RHS of \eq{SqPhimeq1} can be calculated explicitly; in the general case, for $\xi$ of interest,
$\Phi(a;\xi)$ has to be calculated numerically.

\subsubsection{CIR model}\label{ss:CIRmeq1}
The infinitesimal generator is 
\bbe\label{eq:CIRmeq1}
\cL=\frac{\sg^2}{2}x\dd^2_x+\ka(\theta-x)\dd_x+L_{\mathrm{jump}}(D)-x,
\ee
hence, 
\bbe\label{CIRPhimeq1}
\Phi(a;\xi)=\int_{\Ga(a,\xi)}\frac{L_{\mathrm{jump}}(\eta)}{-\sg^2\eta^2/2+i\ka\eta-1}d\eta.
\ee
Let $L_{\mathrm{jump}}(\xi)$ be analytic in a strip $(\lm,\lp)$, where $\lm<-1<0<\lp$. Then,
to price put options, we take $a, \xi$ and a contour $\Ga(a,\xi)$ in the strip $\{\Im\xi\in (0,\lp)\}$,
and to price call options, in the strip $\{\Im\xi\in (\lm,-1)\}$. In simple cases such as exponential jumps
and in cases when $L_J$ becomes a rational function after a polynomial change of variables,
the integral on the RHS of \eq{CIRPhimeq1} can be calculated explicitly; in the general case, for $\xi$ of interest,
$\Phi(a;\xi)$ has to be calculated numerically.

\subsection{The case $m>1$}
We impose an additional condition, which is satisfied in multi-factor Ornstein-Uhlenbeck models, and,
after the conjugation with an appropriate exponential function, in affine term structure models of type $A_0(m)$.  \begin{lem}\label{lem_Phi}  Let condition $ChF(T)$ holds for any $T>0$, and let 
there exist $C, \de>0$ such that
\bbe\label{inf_Asxi}
|A(T,\xi)|\le C(1+|\xi|) e^{-\de T},\ \forall \xi\in\cU_{ChF}, T>0.
\ee
 Then the function
\bbe\label{eq_Phi}
\Phi(\xi)=-i\int_{0}^{+\infty} L_J(-iA(s,\xi))ds
\ee
satisfies \eq{eqPhi1}.
\end{lem}
\begin{proof}
We rewrite \eq{eqPhi1} in the form
\bbe\label{eqPhi2}
\sum_{j=1}^m (-iL_j(\xi))\dd_{\xi_j}\Phi(\xi)=-iL_{J}(\xi),
\ee
change the variable $\xi=-iA$:
\bbe\label{eqPhi3}
\sum_{j=1}^m L_j(-iA))\dd_{A_j}\Phi(-iA)=-iL_{J}(-iA),
\ee
 define $A$ as the solution of the system \eq{eq:A} subject to $A(0,\xi)=i\xi$, and rewrite \eq{eqPhi3} as
\[
\frac{\dd\Phi(-iA(s,\xi))}{\dd s}=-iL_{J}(-iA(s,\xi)).
\]
Integrating the last equality, we obtain \eq{eq_Phi}.

\end{proof} 
\begin{example}{\rm 
 Consider a multi-factor Ornstein-Uhlenbeck driven model, with jumps.
The operator $\cL^0$ is of the form
\bbe\label{H0}
 \cL^0=\frac{1}{2}\langle\Sigma^T\dd,\Sigma^T\dd\rangle+\langle\ka(\theta- x),
 \dd\rangle-\langle d,x\rangle-d_0,
 \ee
 where $\theta, d\in\bR^m$, $d_0\in \bR$, $\Sigma$ is an invertible
$m\times m$ matrix, and all the eigenvalues of the matrix $\ka_0$
 have positive real parts. Equivalently, $\ka$ is  anti-stable  or
 $-\ka$ is  Hurwitz stable. 
 After the conjugation with $\exp[-k'd]$, where $\ka'$ is the transpose to $\ka$, we reduce to the case 
 $d=0$. If $d=0$, equation \eq{eqPhi1} is a linear system
 \bbe\label{eqPhi1OU}
\langle \ka'\xi, \nabla \Phi(\xi)\rangle=L_{J}(\xi).
\ee
The system \eq{eq:A} is a linear system
\bbe\label{eq:A:OU}
\frac{dA}{dt}=-\ka'A,
\ee
which admits an explicit solution $A(t,\xi)=\exp[-t\ka']i\xi$.
}
\end{example}
The conjugation with an appropriate exponential function $x\mapsto e^{\langle a_\infty,x\rangle}$ (equivalently,
an appropriate change of measure) can be used in
more general situations. 
\begin{thm}\label{thm_Phi_gen}
Let there exist $a_\infty\in \bR^m$ and $C, \de>0$ such that
$L_j(-ia_\infty)=0, j=1,2,\ldots, m$,
and,  for all $\xi\in\cU_{ChF}$, the solution of the system
\bbe\label{eq:A:infty}
\frac{dA_j}{dt}=L_j(-i(A(t,\xi)+a_\infty)),\ j=1,2,\ldots,m,
\ee
subject to $A(0,\xi)=i\xi$,  satisfies \eq{inf_Asxi}.
Then 
\bbe\label{eq:Phi:gen}
\Phi(a_\infty;\xi)=-i\int_0^{+\infty}(L_J(-i(A(s,\xi)+a_\infty))-L_J(-ia_\infty))ds
\ee
satisfies
\bbe\label{eqPhi1:ainf}
\sum_{j=1}^m L_j(\xi-ia_\infty)\dd_{\xi_j}\Phi(a_\infty;\xi)=L_{J}(\xi-ia_\infty)-L_J(-ia_\infty),
\ee
and 
\bbe\label{gauge2:gen}
e^{-i\Phi(D)}e^{-\langle a_\infty, x\rangle}(\cL^0+L_J(D))e^{\langle a_\infty, x\rangle}e^{i\Phi(D)}=
\sum_{\ell=1}^mx_\ell L_\ell(D-ia_\infty)+L_0(D-ia_\infty)+L_J(-ia_\infty).
\ee
\end{thm}

\section{Eigenfunction expansions and subordination}\label{Eigenexpansions}
 The first applications of the eigenfunction expansion approach to finance can be found in
\cite{Lewis98, Linetsky99, LiptonLittle, pelsser, LiptonFX, DavidovLinetsky, GorovoiLinetsky,AlbaneseKuznetsov}. Additional applications are in
\cite{LiptonFX,ninaEigenMF,ninaAsymp,Linetsky06,Linetsky08,LiptonGalLasis, LiLinetsky2013,LiLinetsky2014a,LiLinetsky2014b,LiLinetsky2015,LiQuZhang2015}; see also  the bibliographies therein.
In the papers listed above, with the exception of \cite{ninaEigenOU,ninaEigenMF,ninaAsymp},  one-factor models are considered.
Multi-factor Ornstein-Uhlenbeck driven affine and quadratic terms structure models are used in \cite{ninaEigenOU,ninaEigenMF}. The interesting feature of the multi-factor
case is that the models with non-diagonalizable infinitesimal generators $\cL$ fit  the data much better than the ones with
diagonalizable $\cL$. To the best of our knowledge,
 \cite{ninaEigenOU,ninaEigenMF}  are the only papers where the generalized eigenfunction expansion is calculated in multi-factor models,
 including settings with
 non-selfadjoint and non-diagonalizable operators $\cL$. The pricing formulas  involve series in terms of the eigenvalues and (generalized) eigenfunctions,
therefore, the method is practical only in situations where
 explicit formulas for the eigenvalues and eigenfunctions can be derived. If an efficient eigenfunction expansion
 is available in a diffusion model, one can construct a model with jumps using subordination and trivially modifying
 the eigenfunction expansion in the diffusion model. See \cite{AlbaneseKuznetsov,ninaEigenOU,ninaEigenMF,LiLinetsky2014a}
 and the bibliographies therein. For options and CDSs of long maturities,
efficient eigenfunction expansions can be derived in QTSMs perturbed by small jump components \cite{ninaAsymp}.
To the best of our knowledge, \cite{LiMendozaMitchell2016} is the only paper where the closed-form pricing formulas in a jump-diffusion model were derived using the eigenfunction approach; however, the model in  \cite{LiMendozaMitchell2016} is of a rather
special kind: CIR model with embedded exponentially distributed jumps, and the method strongly relies on explicit formulas
available in this model.
The gauge transformation in the dual space allows one to derive efficient eigenfunction expansions for wide classes
of jump-diffusion models and subordinated versions of these jump-diffusion models. 

\subsection{Eigenfunction expansion}\label{ss:eigenfunction expansion}
We want to apply the eigenfunction expansion method to the operator $\cL$ given by
\eq{infgen0} assuming that the eigenvalues of $\cL^0$ given by \eq{cL0} and a basis in an appropriate 
weighted $L_2$-space, consisting of eigenfunctions and generalized eigenfunctions of $\cL^0$, are known or can be calculated with a sufficient precision. 
 
The properties formulated in conditions $\cL 1$-$\cL 3$ below are the properties of strongly elliptic differential and
pseudo-differential operators (certain types of regular degeneration are admissible), which were extensively studied in PDE literature
since the 1960s. See \cite{eskin,Hormander, AsDistr,LevPan,DegEllEq}.
 Define $\cD_0=\bR^{m-m'}\times \bR_{++}^{m'}$.
\vskip0.1cm\noindent
{\sc Condition $\cL 1$}.
There exists a positive function $w\in C(\cD_0)$ and $z_0\in \bC$ such that $\cL^0-z_0$ is an invertible
operator in the $L_2$-space $L_2(w;\cD)$ with weight $w$, and the inverse is a compact operator in $L_2(w;\cD)$.
\vskip0.1cm
One of the main results of Riesz's theory (the theory is valid not only for operators in Hilbert spaces: see any textbook on Functional Analysis)
is that, under condition $\cL 1$, $L_2(w;\cD)$ admits the decomposition into a countable direct sum of finite-dimensional subspaces $V_n$ invariant under
$(\cL^0-z_0)^{-1}$, equivalently, under $\cL^0$:
$L_2(w;\cD)=\oplus_{n=0}^{\infty} V_n$.  Let $P_n: L_2(w;\cD)\to V_n$ be the Riesz's projection on $V_n$.
Denoting by $\cL^0_{n}=P_n \cL^0$ the restriction of
$\cL^0$ to $V_n$, we have the representation of $\cL^0$ as the direct sum of operators acting in finite-dimensional spaces:
$\cL^0=\oplus_{n=0}^{\infty}\cL^0_{n}$. Assuming that $V_n$ are chosen of the smallest dimensions possible, either ${\rm dim}\, V_n=1$
and $-\cL^0_{n}$ is the multiplication-by-$\la_n(\in \bC)$-operator, or ${\rm dim}\, V_n>1$ and $-\cL^0_{n}$
is a Jordan block with the eigenvalue $\la_n\in \bC$. 
According to Riesz's theory, the eigenvalues of $(\cL^0-z_0)^{-1}$ accumulate to 0 only, hence, $\la^0_n\to \infty$ as $n\to +\infty$.

\vskip0.1cm\noindent
{\sc Condition $\cL 2$}. $\Re \la_n\to+\infty$ as $n\to+\infty$.
\vskip0.1cm\noindent
Consider the European option $V^0(G; T;x)$ with the payoff $G(X_T)$ at maturity date $T$, at time 0 and $X_0=x$,
in the model with the infinitesimal generator $\cL^0$. 

\begin{lem}\label{lem:series_eigen0}
Let $G\in L_2(w;\cD)$ and conditions $\cL 1-\cL 2$
hold. 
Then   
\bbe\label{priceX0}
V^0(G; T; x)=(e^{T\cL^0}G)(y)=\sum_{n=0}^\infty (e^{T\cL^0_{n}}G_n)(x),
\ee
where $G_n=P_nG$, and the series converges in $L_2(w;\cD)$-topology. Furthermore,
for any positive integer $m$, the series 
\[
\cL^m\left(\sum_{n=0}^\infty (e^{T\cL^0_{n}}G_n)(x)\right)=\sum_{n=0}^\infty ((\cL^0)^me^{T\cL^0_{n}}G_n)(x)
\]
converges in $L_2(w;\cD)$-topology as well.
\end{lem}
In numerical calculations, the infinite sum above is truncated:
\begin{equation}\label{priceXN}
V^0(G; T; x)\approx \sum_{n=0}^N (e^{T\cL^0_{n}}G_n)(x).
\end{equation}
For applications in engineering and finance, it is important to know that the (generalized) eigenfunction expansion of a function (in finance, price)
 converges not in the $L_2(w;\cD)$-norm only but point-wise as well. Moreover, for the stability of
 numerical calculations, it is important to know that the series \eqref{priceX0} converges in $C(\cD_0)$-topology.
 \footnote{Recall that $u_n\to u$ in $C^s(U)$ as $n\to\infty$, where $U$ is an open set, if and only if for any compact $K\subset U$,
 $u_n|_K\to u|_K$ in $C^{(s)}(K)$ topology. Here, $u|_K$ denotes the restriction of $u$ to $K$.}
  If the eigenfunction expansion is applied to calculate sensitivities,
  one needs to know that the expansion converges in $C^s(\cD_0)$-topology, where $s=1$  (resp., $s=2$) if option's delta
 (resp., gamma) is calculated. For one-factor models, conditions for the point-wise convergence of the price
 can be found in \cite{LiLinetsky2013,LiLinetsky2015}. 
 
 \vskip0.1cm\noindent
{\sc Condition $\cL 3$}. For $s(=0,1,2)$  fixed, there exists a positive integer $m$ such that, for any compact $K\subset \cD_0$,
\begin{equation}\label{apriori}
\|u\|_{C^s(K)}:=\sum_{j=0}^s\sup_{x\in K}|u^{(s)}(x)|\le C_{s,K}(\|(-\cL^0)^m(s)u\|_{L_2(w;\cD)}+\|u\|_{L_2(K)}),
\end{equation}
where the constant $C_{s,K}$ is independent of $u\in L_2(w;\cD)$ for which the RHS is finite.
 \begin{lem}[\cite{BarrStIR}, Lemma 2.1]\label{lemconvCs}
 Let conditions $\cL 1-\cL 3$ hold. Then the series \eqref{priceX0} converges in  $C^s(\cD_0)$-topology.
 \end{lem}
For popular models, including multi-factor ones, Lemma \ref{lemconvCs}
holds for any $s$, hence, the generalized  eigenfunction expansion converges in the
 $C^\infty(\cD_0)$ topology. In \cite{AsDistr,LevPan,DegEllEq,LevConsist1,LevConsist2}, the reader can find analogs of
the bound \eqref{apriori} in stronger topologies that take into account the behaviour at infinity
 and at the boundary of the state space. For instance, in the case of quadratic term structure models,
  for any pair of multi-indices $\al, \be$ there exists $m$ such that
 \begin{equation}\label{aprioriQTSM}
\|y^\be\dd^\al u\|_{L_2(\bR^n)}\le C_{\al,\be}(\|(-\cL^0)^mu\|_{L_2(\bR^n)}+\|u\|_{L_2(\bR^n)}),
\end{equation}
and in the CIR model,
\begin{equation}\label{aprioriCIR}
\|yu^{\prime\prime}\|_{L_2(\bR_+)}+\|(1+y)u'\|_{L_2(\bR_+)}\le C(\|(-\cL^0)u\|_{L_2(\bR_+)}+\|u\|_{L_2(\bR_+)}).
\end{equation}
For similar bounds for some $A_n(m)$ models and outline of the proof
for all $A_n(m)$ models, with jumps, see \cite{LevConsist1,LevConsist2}.

We can use essentially the same decomposition for the operator $\cL=\cL^0+L_J(D)$ if the following
condition holds.
\vskip0.1cm\noindent
{\sc Condition $\cL J$}. $G_\Phi:=e^{-i\Phi(D)}G\in L_2(w;\cD)$, and there exists $m\ge 0$ s.t. $e^{i\Phi(D)}(\cL^0)^{-m}$
is bounded in $L_2(w;\cD)$.

\vskip0.1cm\noindent
Consider the European option $V(G; T;x)$ with the payoff $G(X_T)$ at maturity date $T$, at time 0 and $X_0=x$,
in the model with the infinitesimal generator $\cL$. Using \eq{gauge2} and Lemma \ref{lem:series_eigen0}.
we derive

\begin{lem}\label{lem:series_eigen}
Let $G\in L_2(w;\cD)$ and conditions $\cL 1, \cL 2, \cL J$
hold. 

Then   (a)
\bbe\label{priceX}
V(G; T; x)=e^{i\Phi(D)}(e^{T\cL^0}G_\Phi)(y)=\sum_{n=0}^\infty (e^{i\Phi(D)}e^{T\cL^0_{n}}G_{\Phi,n})(x),
\ee
where $G_{\Phi,n}=P_nG_\Phi$, and the series converges in $L_2(w;\cD)$-topology. 

(b) Let $m>0$, and let there exists $m_1$ such that $(\cL^0)^m e^{i\Phi(D)}(\cL^0)^{-m_1}$
is bounded in $L_2(w;\cD)$. Then the series 
\[
(\cL^0)^m e^{i\Phi(D)}\left(\sum_{n=0}^\infty (e^{T\cL^0_{n}}G_{\Phi,n})(x)\right)
=\sum_{n=0}^\infty ((\cL^0_n)^me^{i\Phi(D)} e^{T\cL^0_{n}}G_{\Phi,n})(x)
\]
converges in $L_2(w;\cD)$-topology as well.
\end{lem}
\begin{rem}{\rm Part (b) and general apriori estimates of the theory of degenerate elliptic equations
\cite{DegEllEq} can be used to establish the convergence of the series on the RHS of \eq{priceX}
in topologies significantly stronger than the one of $L_2(w;\cD)$. In particular, one can derive uniform bounds for
the remainder up to the boundary, as in the square root model and CIR model. We leave the detailed
study of the convergence in stronger topologies for the future.}
\end{rem}

\subsection{Numerical realization}\label{ss:eigen_numer}
For the sake of brevity, we restrict ourselves to the case of completely diagonalizable $\cL^0$'s.
The modification to the case when some of the eigenspaces $V_n$ are Jordan blocks can be done
similarly to   \cite{ninaEigenOU,ninaEigenMF}. Let $\{u_n\}_{n\ge 0}$ be a basis in $L_2(w;\cD)$ of eigenfunctions
of $\cL^0$, and let $\la_n, n\ge 0,$ be the corresponding eigenvalues. The expansion \eq{priceX}
simplifies
\bbe\label{priceXsimple}
V(G; T; x)=\sum_{n=0}^\infty c_n(G_\Phi)e^{-T\la_{n}}(e^{i\Phi(D)}u_n)(x),
\ee
where 
\bbe\label{cnPhi}
c_n(G_\Phi)=(G_\Phi,u_n)_{L_2(w;\cD)}=\int_{\cD} G_\Phi(x)\overline{wu_n(x)}dx.
\ee
\begin{rem}{\rm We write $\overline{wu_n}$ because in multi-factor models, eigenfunctions
can be complex-valued \cite{ninaEigenOU,ninaEigenMF}.
}
\end{rem}
Applying Parceval's equality, we obtain
\begin{lem}\label{lem:calc_cnPhi}
Let there exist $\om\in\cU_G\cap \cU_{ChF}$ such that the functions
$x\mapsto e^{\langle \om,x\rangle}G_\Phi(x)$ and 
$x\mapsto e^{-\langle \om,x\rangle}w(x)u_n(x)$ are in $L_2(\cD)$.
Then 
\bbe\label{cnPhidual}
c_n(G_\Phi)=(G_\Phi,u_n)_{L_2(w;\cD)}=(2\pi)^{-m}\int_{\Im\xi=\om} e^{-i\Phi(\xi)}\hG(\xi)\overline{\widehat{wu_n}(\xi)}d\xi.
\ee
\end{lem}  
Assuming that there exists $\om'\in\cU_{ChF}$ s.t. $e^{-i\Phi}\widehat{u_n}\in L_1(i\om'+\bR^m)$, the ``twisted" eigenfunctions
$e^{i\Phi(D)}u_n$ can be calculated using the Fourier inversion
\bbe\label{twistedPhiun}
e^{-i\Phi(D)}u_n(x)=(2\pi)^{-m}\int_{\Im\xi=\om'}e^{i\langle x,\xi\rangle}e^{-i\Phi(\xi)}\widehat{u_n}(\xi)d\xi.
\ee
In a number of popular models, $\widehat{u_n}(\xi)$ and $\widehat{wu_n}(\xi)$  can be explicitly calculated. Hence, once an efficient procedure for the
evaluation of $\Phi(\xi)$ is designed, the coefficients $c_n(G_\Phi)$ and
the ``twisted" eigenfunctions
$e^{-i\Phi(D)}u_n$ can be calculated accurately and fast using an
appropriate conformal change of variables and simplified trapezoid rule, as explained in \cite{SINHregular,Contrarian}.

In some popular models,  the eigenfunctions $u_n$ are polynomials, hence, $\hu_n$ are linear combinations of
Dirac's delta function $\de$ and its derivatives $\de^{(n)}$, $n=1,2\ldots$. We consider $\de^{(n)}$ as linear continuous 
functionals over 
 $\cS(\bR)$, the class of functions decaying at infinity faster than any polynomial, together with all their derivatives,
and recall that $\de^{(n)}:\cS(\bR)\to \bC$ is defined by 
$(\de^{(n)}, \phi)=(-1)^n \overline{\phi^{(n)}(0)}$, where $(\cdot,\cdot)$ denotes the pairing of a generalized function
and function $\phi\in\cS(\bR)$. See, e.g., \cite{eskin}.
We also need $\cS_\Phi(\bR)\subset \cS(\bR)$, the subspace of $\phi\in \cS(\bR)$ s.t. $\left(e^{i\Phi(D)}\right)^*\phi\in \cS(\bR)$.
(Recall that  the symbol of the operator adjoint to $e^{i\Phi(D)}$ is $\overline{e^{i\Phi(\xi)}}=e^{i\overline{\Phi(\xi)}}$).
If $e^{i\Phi}$ and its derivatives increase not faster than polynomially, then $\cS_\Phi(\bR)= \cS(\bR)$.

\begin{lem}\label{lem: Phixn} For $n\in\bZ_+$,

(a) $\cF x^n=2\pi i^n \de^{(n)} $;

(b) as  a generalized function over $\cS_\Phi(\bR)$,
\bbe\label{ePhixn}
e^{i\Phi(D)}x^n=\sum_{\ell=0}^n (-i)^{n-\ell}{n\choose \ell}\left(\dd_\xi^{n-\ell}e^{i\Phi(\xi)}\right)\vert_{\xi=0} x^\ell.
\ee
\end{lem}

\begin{proof} (a) We have to trivially modify the proof in \cite[Example 2.2]{eskin} because we use 
a marginally different definition
of the Fourier transform. For $\phi\in \cS(\bR)$, we have 
\beqast
(\cF x^k, \hat\phi)&=&2\pi (x^k,\phi)=2\pi\int_\bR x^k\overline{\phi(x)}dx
=2\pi\left(\overline{\int_\bR x^k e^{-ix\xi}\phi(x)dx}\right)\vert_{\xi=0}\\
&=&2\pi (-i)^n \left(\overline{\dd^n_\xi\int_\bR e^{-ix\xi}\phi^{(n)}(x)dx}\right)\vert_{\xi=0}=2\pi i^n \overline{\hat\phi^{(n)}(0)}.
\eqast
(b)  Applying (a), for any $\phi\in \cS_\Phi(\bR)$, we have
\beqast
\left(e^{i\Phi(D)}x^n,\phi\right)&=&\left(x^n, \left(e^{i\Phi(D)}\right)^*\phi\right)\\
&=&(2\pi)^{-1}\left(\cF x^n, e^{-i\overline{\Phi(\cdot)}}\hat\phi(\cdot)\right)\\
&=&(-i)^n\sum_{\ell=0}^n {n\choose \ell}\overline{\left(\dd_\xi^{n-\ell}e^{-i\overline{\Phi(\xi)}}\right)}\vert_{\xi=0}
\overline{\hat\phi^{(\ell)}(0)}\\
&=&\sum_{\ell=0}^n (-i)^{n-\ell}{n\choose \ell}\overline{\left(\dd_\xi^{n-\ell}e^{-i\overline{\Phi(\xi)}}\right)}\vert_{\xi=0}
(2\pi)^{-1}(\cF x^\ell, \phi)
\eqast
and \eq{ePhixn} follows.
\end{proof} In the square root model and CIR model, the eigenfunctions are polynomials which are defined on $\bR_+$ rather than
on $\bR$. In these cases, instead of Lemma \ref{lem: Phixn}, we use the following simple lemma.
\begin{lem}\label{lem: PhixnCIR} For $\al>-1$ and $\om<0$,
$\cF((\bfo_{\bR_+}(x)x^{\al})(\xi)=(i\xi)^{-\al-1}\Ga(\al+1).$
\end{lem}

\subsection{Examples, $m=1$}\label{ss:numer eigen: example onefactor}

\subsubsection{Ornstein-Uhlenbeck model}\label{ss:Oumeq1eigen}
The affine change of variables $x=\frac{\sg}{\sqrt{2\ka}}x'+\theta$ 
reduces to the case $\cL^0=\ka \cL^{00}$, where 
$\cL^{00}u=u^{\prime\prime}-u'$. The operator $-\cL^{00}$ is a self-adjoint operator 
in $L_2(w;\bR)$, where $w(x)=e^{-x^2/2}$, with the discrete spectrum $\bZ_+$,
hence, the spectrum of $-\cL^0$ is $\ka\bZ_+$. Each eigenvalue is simple,
and the corresponding eigenfunctions are proportional to Hermite polynomials:
\bbe\label{unOU}
u_n(x')=H_n(x'):=(-1)^ne^{(x')^2/2}\dd_{x'}^ne^{-(x')^2/2}=(x'-\dd_{x'})^n\bfo.
\ee
 The twisted eigenfunctions $e^{i\Phi(D)}u_n$ can be calculated using Lemma \ref{lem: Phixn},
 and the calculation of the Fourier transform of $wu_n$ is straightforward:
 \beqa\label{OUwun}
 \widehat{wu_n}(\xi)&=&(-1)^n\int_\bR e^{-ix\xi} \dd_x^n e^{-x^2/2}dx\\\nonumber
 &=&(-i\xi)^n \int_\bR e^{-ix\xi}e^{-x^2/2}dx\\\nonumber
 &=&(-i\xi)^n e^{-\xi^2/2}\int_\bR e^{-(x+i\xi)^2/2}dx\\\nonumber
 &=&\sqrt{2\pi}(-i\xi)^n e^{-\xi^2/2}.
 \eqa
 
 \subsubsection{Vasicek model}\label{ss:Vasmeq1eigen}
 We conjugate $\cL^0=\frac{\sg^2}{2}\dd^2_x+\ka(\theta-x)\dd_x-x$ with $e^{-ax}$, where $a=1/\ka$, to eliminate the term $-x$:
 \beqast
\cL^0_\ka&:=& e^{ax}\cL^0 e^{-ax}\\
&=&\frac{\sg^2}{2}(\dd_x-a)^2+\ka(\theta-x)(\dd_x-a)-x\\
 &=&\frac{\sg^2}{2}\dd^2_x+(\ka\theta-a\sg^2)\dd_x-\ka x\dd_x+\frac{\sg^2a^2}{2}+a\ka\theta\\
 &=&\frac{\sg^2}{2}\dd^2_x+\ka(\theta_1-x)\dd_x+\La_0,
 \eqast
where $\theta_1=\theta-(\sg/\ka)^2$ and $\La_0=\frac{\sg^2}{2\ka^2}+\theta$. 
The affine change of variables $x=\frac{\sg}{\sqrt{2\ka}}x'+\theta_1$ 
reduces to the case $\cL^0_\ka=\ka \cL^{00}+\La_0$, where 
$\cL^{00}u=u^{\prime\prime}-u'$. Hence, the spectrum of $-\cL^0_\ka$ is $\La_0+\ka\bZ_+$,
and eigenfunctions (in the new coordinates) are the Hermite polynomials $H_n(x')$. To apply the twisted eigenfunction expansion,
we add to the calculations of Sect.~\ref{ss:Oumeq1eigen} the following additional first and  last steps:
at the beginning, multiply $G$ by $e^{x/\ka}$ (not by $e^{x'/\ka}$); in the end, multiply the eigenfunction expansion by 
$e^{-x/\ka}$.

 \subsubsection{Square root model}\label{ss:Sqmeq1eigen}
 The change of variables $x=x'\sg^2/(2\ka)$ reduces  $\cL^0$ to
 the case $\cL^0=\ka\cL^{00}$, where $\cL^{00}=x'\dd_{x'}^2+(1+\al-x')\dd_{x'}$, and $\al=2\ka\theta\sg^{-2}-1\in\bR$. 
 Assume that $\cL^{00}$ satisfies the Feller condition. Then $-\cL^{00}$ is a self-adjoint operator
 in $L_2(w;\bR_+)$, where $w(x')=(x')^\al e^{-x'}$, with the discrete spectrum $\bZ_+$. Each eigenvalue  $\la_n=n, n\in\bZ_+,$ is simple, and the corresponding eigenfunctions are the generalized Laguerre polynomials
  \[
  u_n(x'):=L_n^{(\al)}(x')=\frac{1}{n!}x^{-\al}(\dd_{x'}-1)^n (x')^{n+\al}.
  \]
  Hence, the twisted eigenfunctions $e^{i\Phi(D)}u_n$ can be calculated using Lemma \ref{lem: PhixnCIR},
 and the calculation of the Fourier transform of $wu_n$ is straightforward. Since
 $w(x)u_n(x)=e^{-x}(\dd_x-1)^nx^{n+\al}=\dd_x^n e^{-x}x^{n+\al}$, we have
 \beqa\label{SqRwun}
\widehat{wu_n}(\xi)&=&(i\xi)^n\cF(e^{-x}x^{n+\al})\\\nonumber
&=&(i\xi)^n\int_{\bR_+}e^{-ix\xi}e^{-x}x^{n+\al}dx\\\nonumber
&=&(i\xi)^n(1+i\xi)^{-n-\al-1}\Ga(n+\al+1).
\eqa

\subsubsection{CIR model}\label{ss:CIRmeq1eigen}
Define $a=-(\ka\theta+\sqrt{(\ka\theta)^2+2\sg^2})/\sg^2$.
We conjugate $\cL^0=\frac{\sg^2}{2}x\dd^2_x+\ka(\theta-x)\dd_x-x$ with $e^{-ax}$, and, taking into account that
$\sg^2a^2/2+\ka\theta a-1=0$, obtain:
 \beqast
\cL^0_a&:=& e^{ax}\cL^0 e^{-ax}\\
&=&\frac{\sg^2}{2}x(\dd_x-a)^2+\ka(\theta-x)(\dd_x-a)-x\\
 &=&\frac{\sg^2}{2}x\dd^2_x+\ka\theta\dd_x-(\ka+a\sg^2) x\dd_x\\
 &=&\frac{\sg^2}{2}x\dd^2_x+\ka_1(\theta_1-x)\dd_x,
 \eqast
where $\ka_1=\ka+a\sg^2$ and $\theta_1=\theta\ka/\ka_1$. Making the change of variables 
$x=\sg^2/(2\ka_1) x'$,  one reduces  $\cL^0_a$ to
 the case $\cL^0_a=\ka_1\cL^{00}$, where $\cL^{00}=x'\dd_{x'}^2+(1+\al-x')\dd_{x'}$, 
 and $\al=2\ka_1\theta_1\sg^{-2}-1\in\bR$.  Assume that $\cL^{00}$ satisfies the Feller condition. Then
 the spectrum of $-\cL^0_a$ is $\ka_1\bZ_+$,
and eigenfunctions (in the new coordinates) are generalized Laguerre polynomials
 $L^{(\al)}_n(x')$. To apply the twisted eigenfunction expansion,
we add to the calculations of Sect.~\ref{ss:Sqmeq1eigen} the following additional first and  last steps:
at the beginning, multiply $G$ by $e^{ax}$ (not by $e^{ax'}$); in the end, multiply the eigenfunction expansion by 
$e^{-ax}$.

\subsection{Example, $m>1$: multi-factor Ornstein-Uhlenbeck driven models} The eigenfunctions and generalized eigenfunctions are expressed in terms of Hermite polynomials \cite{ninaEigenOU,ninaEigenMF}, hence,
the numerical realization is similar to the one in Sect.~\ref{ss:Oumeq1eigen} and \ref{ss:Vasmeq1eigen}. 

\subsection{Subordinated affine jump-diffusion processes}\label{subord}
A subordinator is a L\'evy process taking values in $[0,
+\infty)$, which implies that its trajectories are non-decreasing.
The Laplace transform of the law of the subordinator $Z$ can be
expressed as $\bE[\exp(-\la Z_t)]=\exp(-t\Psi(\la))$, where
$\Psi:\bR_+\to \bR_+$ is called the Laplace exponent of $Z$. The
Laplace exponent is given by
 \bbe\label{lapexp}
 \Psi(\la)=\ga\la +\int_0^{+\infty}(e^{\la s}-1)F(ds),
 \ee
 where $\ga\ge 0$, and $F(dy)$ is the L\'evy density of $Z$, which satisfies
  \bbe\label{levymeas}
  \int_0^{+\infty} \min\{1, y\}F(dy)<+\infty.
  \ee
  The subordinated process $Y(t)=X(Z(t))$ is a Markov process
  with the infinitesimal generator
   \bbe\label{defsubgen}
   \Psi(\cL)=\ga \cL+\int_0^{+\infty}(e^{s\cL}-1)F(ds).
   \ee
   The backward parabolic equation for the price assumes the form
 \beqa\label{subBK}
  (\dd_t+\Psi(\cL))u(t,x)&=&0,\quad t<T,\\\label{subBKterm}
  u(T,x)&=&G(x),
  \eqa
  equivalently,
  \beqa\label{subBKtwist}
  (\dd_t+e^{-i\Phi(D)}\Psi(\cL)e^{i\Phi(D)})u_\Phi(t,x)&=&0,\quad t<T,\\\label{subBKtermtwist}
  u_\Phi(T,x)&=&G_\Phi(x),
  \eqa
  $\cL^0$ can be decomposed into a direct sum of finite-dimensional operators:
$\cL^0=\oplus_{n\ge 0}\cL^0_n$,  hence,
$e^{s\cL^0}=\oplus_{n\ge 0}e^{s\cL^0_n}$ and $e^{-i\Phi(D)}e^{s\cL}e^{i\Phi(D)}=\oplus_{n\ge 0}e^{s\cL^0_n}$.
Substituting the last equality into \eq{defsubgen}, we conclude that
\bbe\label{subtwist}
e^{-i\Phi(D)}\Psi(\cL)e^{i\Phi(D)}=\oplus_{n\ge 0}\Psi(\cL^0_n).
\ee
It follows that the pricing problem in the subordinated jump diffusion model can be solved
as the one in the non-subordinated model. In particular, if all the eigenvalues are simple,
then the only change in the pricing formula in the non-subordinated model is the replacement of each $e^{-T\la_n}$
with $e^{-T\Psi(\la_n)}$. If the decomposition of $\cL^0$ has Jordan blocks, the changes are somewhat more involved but
still simple. See \cite{ninaEigenOU,ninaEigenMF}.

\section{Estimation of the rare jump component}\label{estimation}
\subsection{Estimation in the long run: general discussion}
Recently, in a number of papers (see
\cite{HansenScheinkman09,HansenScheinkman12a,LQinLinetsky16,LQinLinetsky17,LQinLinetskyNie18}
and the bibliographies therein),
an important problem of the long-term factorization in Markovian models and recovery of
the historic measure from the prices of assets of asymptotically long maturities were studied. In the framework of these papers, the recovery procedure
is possible if the leading eigenfunction of the pricing operator  is recurrent, and the deep theoretical
result useful if the leading eigenfunction can be explicitly inferred from the data.

We consider the closely related problem of identification of a component of extremely rare jumps which 
never happened  in the past but whose possibility the market can anticipate. A natural example
is the pricing of assets  in the markets vulnerable to extreme climate changes and related regulatory shocks.
  Under the following assumptions 
  \begin{enumerate}[(1)]
  \item
   the dynamics of prices
in the market is described by an affine jump-diffusion model; 
\item
the parameters of the diffusion component
and frequent small jump  component can be inferred from the prices of assets of short maturities; 
\item a sufficiently large
number of options of long maturities is traded in the market, 
\end{enumerate}
we derive an explicit formula for the characteristic exponent
of the extremely rare jump component. Hence, the beliefs of the market about yet unobserved extreme jumps
and pricing kernel 
can be recovered.  If we make an assumption on which the popular Carr-Madan 
static hedging procedure \cite{carr-madan-statichedge} is based, namely, that a continuum of
put or call options (for our purposes, of a long maturity $T$) are available in the market, then the error of the recovery procedure that we use decays
exponentially as a function of $T$. 

Note that Assumption (3) is rather strong but it is not as strong as the underlying assumptions
in op.cit.: the process is ergodic, and the leading eigenfunction can be recovered from the prices of assets of long maturities. The recovery procedure that we derive allows one to see ``the shape of things to come" which have not been observed or, rather,
as this shape is seen by the market.

\subsection{Estimation using options on the underlying, in one-factor models} We assume that the dynamics of the underlying follows the process with the infinitesimal generator $\cL$ given by \eq{infgen0}. We assume that the parameters of $L(D)=[L_j(D)]_{j=0}^m$
are inferred from the prices of options of short maturity, and 
we want to infer the parameters of the infinitesimal generator (equivalently, of the characteristic component)
using the prices of options of long maturity. 

We make an assumption similar to the assumption underlying a popular static hedging procedure due to P.~Carr and 
D.~Madan \cite{carr-madan-statichedge}, but do the calculations in the dual space, in the same vein as the static hedging formula in 
\cite{Contrarian} is derived.

\vskip0.1cm\noindent
{\sc Assumption LTK.} For a very large  $T$, there exists a continuum of call options or continuum of 
put options on the underlying, of maturity $T$, with strikes $K\in (0,+\infty)$ if the model is exponential or $K\in (-\infty,+\infty)$, if the model is arithmetic.

\subsubsection{Exponential model} Set $k=\ln K$. For any $\om<-1$ in the case of call options, and
$\om>0$ in the case of put options, the Fourier transform of the option with strike $K$ is well-defined on the line $\Im \xi=\om$:
$\hG(K;\xi)=K^{1-i\xi}\hG_0(\xi)$, where $\hG_0(\xi) =-1/(\xi(\xi+i))$. Assume that $\{\Im\xi=\om\}$ is in the strip of analyticity
of the characteristic function, calculate the function $A$, define $B_0$ by the RHS of \eq{eq:B} and $B_J$ by the 
RHS of \eq{eq:B} with $L_J$ in place of $L_0$. 
Then, in the model with the infinitesimal generator $\cL$,
the option price is
\beqa
\label{priceVK}
V(T;K;x)&=&\frac{K}{2\pi}\int_{\Im\xi=\om} e^{A(T,\xi)x+B_0(T,\xi)+B_J(T,\xi)-ik\xi}\hG_0(\xi)d\xi.
\eqa
Taking the Fourier transform of $V_1(T,k,x):=e^{-k}V(T;e^k;x)$ w.r.t. $k$, we obtain
\beqa
\label{pricehVK}
\hV_1(T;\xi;x)&=&e^{A(T,\xi)x+B_0(T,\xi)+B_J(T,\xi)}\hG_0(\xi).
\eqa
Hence, 
\bbe\label{BJ1}
 B_J(T,\xi)=h(T,\xi):=-A(T,\xi)x-B_0(T,\xi)-\ln\frac{\hV_1(T;\xi;x)}{\hG_0(\xi))}.
 \ee
\begin{rem}{\rm 
\begin{enumerate}[(i)]
\item $A(T,\xi)$, $B_0(T,\xi)$ 
can be calculated exactly since we assume that the parameters of $\cL_0$ are
known (inferred from the prices of options of short maturities).
\item
The spot $X_0=x$ is also presumed to be inferred from the prices of options of short maturities.
\item $\hV(T;\xi;x)$ can be calculated by integration using the
data on options of long maturity $T$.
 Any numerical quadrature uses only the values of the underlying at a finite number of points, hence,
 the estimation procedure can be used in practice where only a finite number of options is available.
 \end{enumerate}
 }
 \end{rem}
 Under condition \eq{inf_Asxi},
 there exists $\de>0$ s.t.
 \bbe\label{BJinf}
 B_J(T,\xi)-B_J(+\infty, \xi)=O\left(e^{-\de T}\right), T\to +\infty.
 \ee
 Since $-iB_J(+\infty,
 \xi)=\Phi(\xi)$, for large $T$, we can use the approximation
 \[
 \Phi(\xi)\approx -ih(T,\xi).
 \]
For a given (large) $T$, the function $\xi\mapsto h(T,\xi)$, hence, $\Phi(\xi)$, has been inferred from the data.
Functions $\xi\mapsto L_j(\xi), j=1,2,\ldots, m$, have been inferred from the data as well.
Using \eq{eqPhi1}, we recover the characteristic function $\psi(\xi)=-L_J(\xi)$ of the yet unobserved jump component.

\subsubsection{Arithmetic model} Set $k=\ln K$. For any $\om<-1$ in the case of call options, and
$\om>0$ in the case of put options, the Fourier transform of the option price the option with strike $K$ is well-defined on the line $\Im \xi=\om$:
$\hG(K;\xi)=e^{-iK\xi}\hG_0(\xi)$, where $\hG_0(\xi) =-\xi^{-2}$. Assume that $\{\Im\xi=\om\}$ is in the strip of analyticity
of the characteristic function, calculate the function $A$, define $B_0$ by the RHS of \eq{eq:B} and $B_J$ by the 
RHS of \eq{eq:B} with $L_J$ in place of $L_0$. Then, in the model with the infinitesimal generator $\cL$,
the option price is
\beqa
\label{priceVKar}
V(T;K;x)&=&\frac{1}{2\pi}\int_{\Im\xi=\om} e^{A(T,\xi)x+B_0(T,\xi)+B_J(T,\xi)-iK\xi}\hG_0(\xi)d\xi.
\eqa
Taking the Fourier transform w.r.t. $K$, we obtain
\beqa
\label{pricehVKar}
\hV(T;\xi;x)&=&e^{A(T,\xi)x+B_0(T,\xi)+B_J(T,\xi)}\hG_0(\xi),
\eqa
and continue as in the case of the exponential models.

\subsection{Estimation using options on forwards or bonds of short maturity, in one-factor models} 
If options of long maturity are not on the underlying but on
forwards or bonds of short maturity, then, in the construction above, the underlying stochastic factor needs to be
replaced by an affine function of the factor. If options available in the market are on forwards or bonds of long maturity,
then the method of the paper requires modifications, which we leave for the future.

\subsection{Estimation in multi-factor models} The estimation procedure described above admits a straightforward modification
to the case of multi-factor models if the prices of a continuum of products of options can be inferred from the observations of
options traded in the market. We leave the study of this possibility for the future. In the multi-factor case, it may be necessary
to make a preliminary change of measure (see Theorem.~\ref{thm_Phi_gen}).





\section{Conclusion}\label{concl}
In the paper, we use the gauge transformation in the dual space (conjugation with an operator
of the form $e^{i\Phi(-i\dd_x)}$) to eliminate a jump component of the infinitesimal generator of the model thereby simplifying
the model. We derive the equation \eq{eqPhi1}
for $\Phi$ and show how this equation can be solved in several classes of models.
If the jump component can be completely eliminated, and the (generalized) eigenfunction expansion in the diffusion model is available, then the basis of generalized eigenfunctions in the diffusion model can be used to derive an explicit series representation of the price (``twisted eigenfunction expansion"). The straightforward modification of 
the procedure in \cite{ninaEigenOU,ninaEigenMF} is used to derive the representation of the price in subordinated jump-diffusion models. 

The second application of equation \eq{eqPhi1} is estimation of the rare jump component; we suggest to call 
equation \eq{eqPhi1} the {\em elimination-estimation equation} ({\em EE equation}). 
For application of nature-industry interaction models,
 it is necessary to infer from the data the dynamics of future extreme shocks as it seen by the market.
 Recently, a number of papers addressed
an important problem of the long-term factorization in Markovian models, and recovery of
the historic measure from the prices of assets of asymptotically long maturities (Hansen-Scheinkman factorization and Ross recovery). In the framework of these papers, the recovery procedure
is possible if the leading eigenfunction of the pricing operator  is recurrent, and the deep theoretical
result useful if the leading eigenfunction can be explicitly inferred from the data.
We consider the closely related problem of identification of a component of extremely rare L\'evy jumps which 
never happened  in the past but whose possibility the market can anticipate. A natural example
is the pricing of assets  in the markets vulnerable to extreme climate changes and related regulatory shocks.
  The beliefs of the market about yet unobserved extreme jumps
and pricing kernel and
can be recovered: the market prices allow one to see  ``the shape of things to come".

\end{document}